\documentclass[10pt]{article}
\usepackage{url}
\usepackage{helvet}
\usepackage{times}
\usepackage{cite}
\usepackage{floatflt}

\usepackage{amssymb}
\usepackage{graphicx}
\usepackage{subfigure}
\usepackage{rotating}
\usepackage{dsfont,graphics,epsfig,psfrag,stmaryrd,amsmath,amsfonts,amssymb,mathrsfs,txfonts,booktabs,times} 
\usepackage{algorithm}
\usepackage{algorithmic}
\usepackage{multirow}
\usepackage[standard]{ntheorem}

\begin{document}

\title{FPGA Design for Pseudorandom Number Generator Based on Chaotic Iteration used in Information Hiding Application
}
\author{Jacques M. Bahi, Xiaole Fang, Christophe Guyeux, Laurent Larger~\thanks{Authors in alphabetic order}}

\maketitle

\begin{abstract}
Lots of researches indicate that the inefficient generation of random numbers is a significant bottleneck for information communication applications. Therefore, Field Programmable Gate Array (FPGA) is developed to process a scalable fixed-point method for random streams generation. In our previous researches, we have proposed a technique by applying some well-defined discrete chaotic iterations that satisfy the reputed Devaney's definition of chaos, namely chaotic iterations (CI). We have formerly proven that the generator with CI can provide qualified chaotic random numbers. In this paper, this generator based on chaotic iterations is optimally redesigned for FPGA device. By doing so, the 
generation rate can be largely improved. Analyses show that these hardware generators can also provide good statistical chaotic random bits and can be 
cryptographically secure too.
An application in the information hiding security field is finally given as an illustrative example.
\end{abstract}

\section{Introduction}
The extremely rapid development of the Internet brings more and more attention to the information security techniques, such as text, image, or video encryption, etc. As a result, highly qualified random sequences, as an inseparable part of encryption techniques, are urgently required. There are two kinds of random sequences: real random sequences generated by physical methods and pseudorandom sequences generated by algorithm simulations, which are in accordance with some kind of probability distributions. The implementation methods for different classes of random number generators are visualized in Figure~\ref{classification}. However, the constructions of the real random sequences are usually poor in speed and efficiency, and require considerably more storage space as well, and these defects restrict their usage in modern cryptography. On the one hand, field programmable gate arrays (FPGAs) have been successfully used for realizing the speed requirement in pseudorandom sequence generation, due to their high parallelization 
capability \cite{Bojani200663,Danger,Tsoi:2003:CFT:938383.938400}. Advantages of such physical generation way encompass performance, design time, power consumption flexibility, and cost. On the other hand, there is a growing interest to use chaotic dynamical systems as PRNGs, among other things due to the unpredictability and distorted-like properties of such systems (\cite{PhysRevE.72.016220,Cecen2009847,Lee2004187}). Nowadays, such chaos-based generators have also been successfully used to strengthen optical communications \cite{cite-key}.
\begin{figure*}
\begin{center}
  \includegraphics[width=12cm]{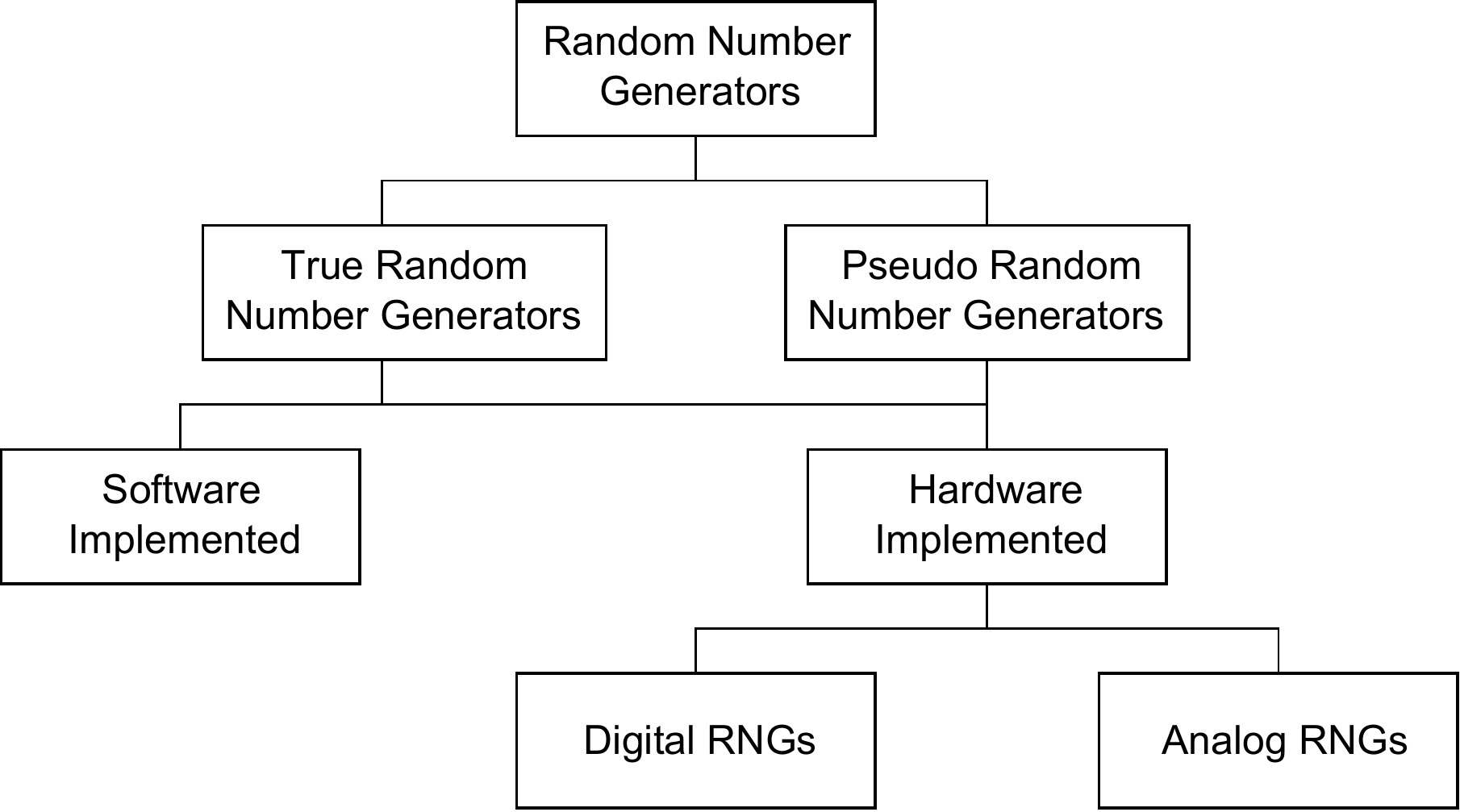}
\end{center}
\caption{Classification of random number generators}
 \label{classification}
\end{figure*} 

A short overview of our previous researches is given thereafter. It has firstly been stated that a tool called chaotic iterations (CIs), used in distributed computing, satisfies the chaotic property as it is defined by Devaney~\cite{citeulike:612765}. The chaotic behavior of CIs has then been exploited to obtain a class of unpredictable PRNGs~\cite{DBLP:journals/corr/abs-1004-2374}. This class receives two given, potentially defective, generators as input and mix them with chaotic iterations, producing by doing so a sequence having a better random profile than the two inputs taken alone~\cite{bfg12a:ip}. Then, in~\cite{bfgw11:ij}, two new versions of such ``CIPRNGs'' have been proposed, involving respectively two logistic maps 
and two XORshifts. 

In this paper, we continue the works initiated in \cite{DBLP:journals/corr/abs-1004-2374,bfgw11:ij,guyeux:topological,bfg12a:ip}: the two approaches introduced before are merged by proposing a discrete chaos-based generator designed on FPGA. The idea is to improve the efficiency of our formerly proposed generators, without any lack of chaos properties. To do so, a new model of CIPRNG Version 1~\cite{DBLP:journals/corr/abs-1004-2374} on Field Programmable Gate Array is introduced and its security is proven in some cases. Additionally, the randomness of this novel proposal is evaluated by the famous NIST test suite (widely used as a randomness standard battery of tests~\cite{ANDREW2008}). Last but not the least, a potential usage of this generator in a cryptographic application is presented. 

\section{Definitions and terminologies}
\label{basic recall}
\subsection{Notations}
\begin{tabular}{@{}c@{}@{}l@{}}
$\llbracket 1;N \rrbracket$ & $\rightarrow\{1,2,\ldots,N\}$ \\
$S^{n}$ & $\rightarrow$ the $n^{th}$ term of a sequence $S=(S^{1},S^{2},\ldots)$ \\
$v_{i}$ & $\rightarrow$ the $i^{th}$ component of a vector: $v=(v_{1},v_{2},\ldots, v_n)$\\
$\emph{strategy}$~ & $\rightarrow$ a sequence which elements belong in $%
\llbracket 1;N \rrbracket $ \\
$\mathbb{S}$ & $\rightarrow$ the set of all strategies \\
$X^\mathds{N}$ & $\rightarrow$ the set of sequences belonging into $X$\\
$\mathbf{C}_n^k$ & $\rightarrow$ the binomial coefficient ${n \choose k} = \frac{n!}{k!(n-k)!}$\\
$+$ & $\rightarrow$ the integer addition \\
$\ll \text{and} \gg$ & $\rightarrow$ the usual shift operators \\
$\mathds{N}^{\ast }$ & $\rightarrow$ the set of positive integers \{1,2,3,...\}\\
$\&$ & $\rightarrow$ the bitwise AND\\
$\oplus$ & $\rightarrow$ the bitwise exclusive or between two integers.
\end{tabular}

\subsection{Blum Blum Shub and XORshift}
The Blum Blum Shub generator~\cite{BBS} (usually denoted by BBS) takes the form:
$$x^{n+1}=\left(x^n\right)^2~ mod~ m, ~~ y^{n+1} = x^{n+1}~ mod~ log(log(m)),$$  
where $m$ is the product of  two prime numbers (these prime numbers  need to be congruent  to 3 modulus 4), and $y^n$ is the returned 
binary sequence. 
\begin{algorithm}
\textbf{Input}: $x$ (a 64-bit word)\\
\textbf{Output}: $r$ (a 64-bit word)\\
\textbf{Parameters}: $a,b,c$ (integers)\\
\begin{algorithmic}[1]
\STATE$x \leftarrow x \oplus (x \ll a);$\\
\STATE$x \leftarrow x \oplus (x \gg b);$\\
\STATE$x \leftarrow x \oplus (x \ll c);$ \\
\STATE$r \leftarrow x;$\\
\STATE\textbf{An arbitrary round of XORshift}~\\
\end{algorithmic}
\caption{XORshift algorithm}
\label{xorshift algorithm}
\end{algorithm}

XORshift, on its part, is a category of very fast PRNGs designed by George Marsaglia \cite{Marsaglia:2003:JSSOBK:v08i14}. 
Algorithm \ref{xorshift algorithm} shows its working procedure. The values of $a,b,c$ decide the offsets of shifting.

\subsection{Chaotic iterations}
\label{subsection:Chaotic iterations}

\begin{definition}
\label{Chaotic iterations}
The set $\mathds{B}$ denoting $\{0,1\}$, let $f:\mathds{B}^{N}\longrightarrow \mathds{B}^{N}$ be an ``iteration'' function and $S\in \mathbb{S}$ be a strategy. Then, the so-called \emph{chaotic iterations} are defined by~\cite{Robert1986}:

\begin{equation}
\left\{\begin{array}{l}
x^0\in \mathds{B}^{N}, \\
\forall n\in \mathds{N}^{\ast },\forall i\in \llbracket1;N\rrbracket%
,x_i^n=
\left\{
\begin{array}{ll}
x_i^{n-1} & \text{if}~S^n\neq i \\
f(x^{n-1})_{S^n}  & \text{if}~S^n=i.
\end{array}
\right.
\end{array}
\right.
\end{equation}
\end{definition}

In other words, at the $n^{th}$ iteration, only the $S^{n}-$th cell is \textquotedblleft iterated\textquotedblright . 
Note that in a more general formulation, $S^n$ can be a subset of components and $f(x^{n-1})_{S^{n}}$ can be replaced by 
$f(x^{k})_{S^{n}}$, where $k < n$, describing for example delays transmission. For the
general definition of such chaotic iterations, see, e.g.,~\cite{Robert1986}.

Chaotic iterations generate a set of vectors (Boolean vectors in this paper), they are defined by an initial 
state $x^{0}$, an iteration function $f$, and a strategy $S$
said to be a ``chaotic strategy''.
Being an iterative process producing binary vectors given a ``seed'' $x^0$, such chaotic iterations 
can be used as pseudorandom number generators. The mathematical fundations of
such a contruction is recalled in the next section.

\subsection{Chaotic iterations as PRNG}

\label{subsec Chaotic iterations as PRNG}
Our generator denoted by $CI_f(PRNG1,PRNG2)$ is de\-si\-gned by the following process. 

Let $\mathsf{N} \in \mathds{N}^*, \mathsf{N} \geqslant 2$. Some chaotic iterations are fulfilled, 
with $f$ as iteration function and $PRNG1$ for strategy, to generate 
a sequence $\left(x^n\right)_{n\in\mathds{N}} \in \left(\mathds{B}^\mathsf{N}\right)^\mathds{N}$ of 
Boolean vectors: 
the successive states of the iterated system. Some of these vectors are randomly extracted using
$PRNG2$, and their 
components constitute our pseudorandom bit flow.

Chaotic iterations are realized as follows. Initial state $x^0 \in \mathds{B}^\mathsf{N}$ is 
a Boolean vector taken as a seed and chaotic strategy 
$\left(S^n\right)_{n\in\mathds{N}}\in \llbracket 1, \mathsf{N} \rrbracket^\mathds{N}$ is 
constructed with PRNG2. Lastly, iterate function $f$ is the vectorial Boolean negation
$$f_0:(x_1,...,x_\mathsf{N}) \in \mathds{B}^\mathsf{N} \longmapsto (\overline{x_1},...,\overline{x_\mathsf{N}}) \in \mathds{B}^\mathsf{N}.$$
To sum up, at each iteration only $S^i$-th component of state $X^n$ is updated, as follows
\begin{equation}
x_i^n = \left\{\begin{array}{ll}x_i^{n-1} & \text{if } i \neq S^i, \\ \\ \overline{x_i^{n-1}} & \text{if } i = S^i. \\\end{array}\right.
\end{equation}

Finally, let $\mathcal{M}$ be a finite subset of $\mathds{N}^*$. Some $x^n$ are selected by a sequence $m^n$ as the pseudorandom bit sequence of our generator, $(m^n)_{n \in \mathds{N}} \in \mathcal{M}^\mathds{N}$ . So, the generator returns the following values: the components of $x^{m^0}$, followed by the components of $x^{m^0+m^1}$, followed by the components of $x^{m^0+m^1+m^2}$, \emph{etc.} In other words, the generator returns the following bits:

$$x_1^{m_0}x_2^{m_0}x_3^{m_0}\hdots x_\mathsf{N}^{m_0}x_1^{m_0+m_1}x_2^{m_0+m_1}\hdots x_\mathsf{N}^{m_0+m_1} x_1^{m_0+m_1+m_2}\hdots$$

\noindent or the following integers:$$x^{m_0}x^{m_0+m_1}x^{m_0+m_1+m_2}\hdots$$

In details, when considering the Boolean negation and two integer 
sequences $p$ and $q$, we obtain the CIPRNG($p$,$q$) 
version 1 published in~\cite{wang2009}: $p$ is $S$ and the output of the generator
is the subsequence $\left(x^{\sigma(n)}\right)_{n\in\mathds{N}}$, where $\sigma(0)=q^0$ and
$\sigma(n+1) = \sigma(n)+q^n$. 
Reason to be of the sequence $q$ is that, between two iterates of chaotic
iterations, at most 1 bit will change in the vector, and thus the sequence
$(x^n)$ cannot pass any statistical test: we must extract a subsequence $(x^{\sigma(n)})$ of $(x^n)$
to produce the outputs.
 CIPRNG($p$,$q$) version 2, for its part, will extract
a subsequence from the strategy $S=p$ to prevent from negating several times a
same position between two outputs.

\begin{example}
If we consider the Boolean negation for $f$, then chaotic iterations 
of Definition~\ref{Chaotic iterations} can be rewritten as: $x^{n+1} = x^n \oplus s^n$, where
$s^n \in \llbracket 0, 2^{\mathsf{N}-1}\rrbracket$ is such that
its $k-$th binary digit is 1 if and only if $k \in S^n$.
Such a particular chaotic iterations will be our generator called 
XOR CIPRNG~\cite{DBLP:journals/corr/abs-1112-5239}.
\end{example}

\subsection{PRNGs based on chaotic iterations}

Let us now recall with more details some previous works in the
field of CIPRNGs: chaotic iteration based 
pseudorandom number generators.

\subsubsection{CIPRNG(PRNG1,PRNG2): Version 1}
Let PRNG1 and PRNG2 be two given generators
provided as input, or ``entropy sources''. The objective of the CIPRNG approach
is to mix them together using chaotic iterations, in such a way
that chaos improve their statistics against well-known batteries of tests, while the speed of the resulted
mixed PRNGs is of the same order than the slowest input.
Additionally, we will show in a further section that if the
PRNG1 is cryptographically secure, then it is the case too
for the mixed CIPRNG(PRNG1,PRNG2).
Thus expected properties of entropy sources could be, for
instance, speed
for PRNG2 and security or good statistical properties for PRNG1,
even though, theoretically speaking, nothing is required for
these inputs except that they must not be totally defective (chaos cannot
correct constant inputs for instance).

\begin{figure}[!t]
\centering
\includegraphics[width=2.5in]{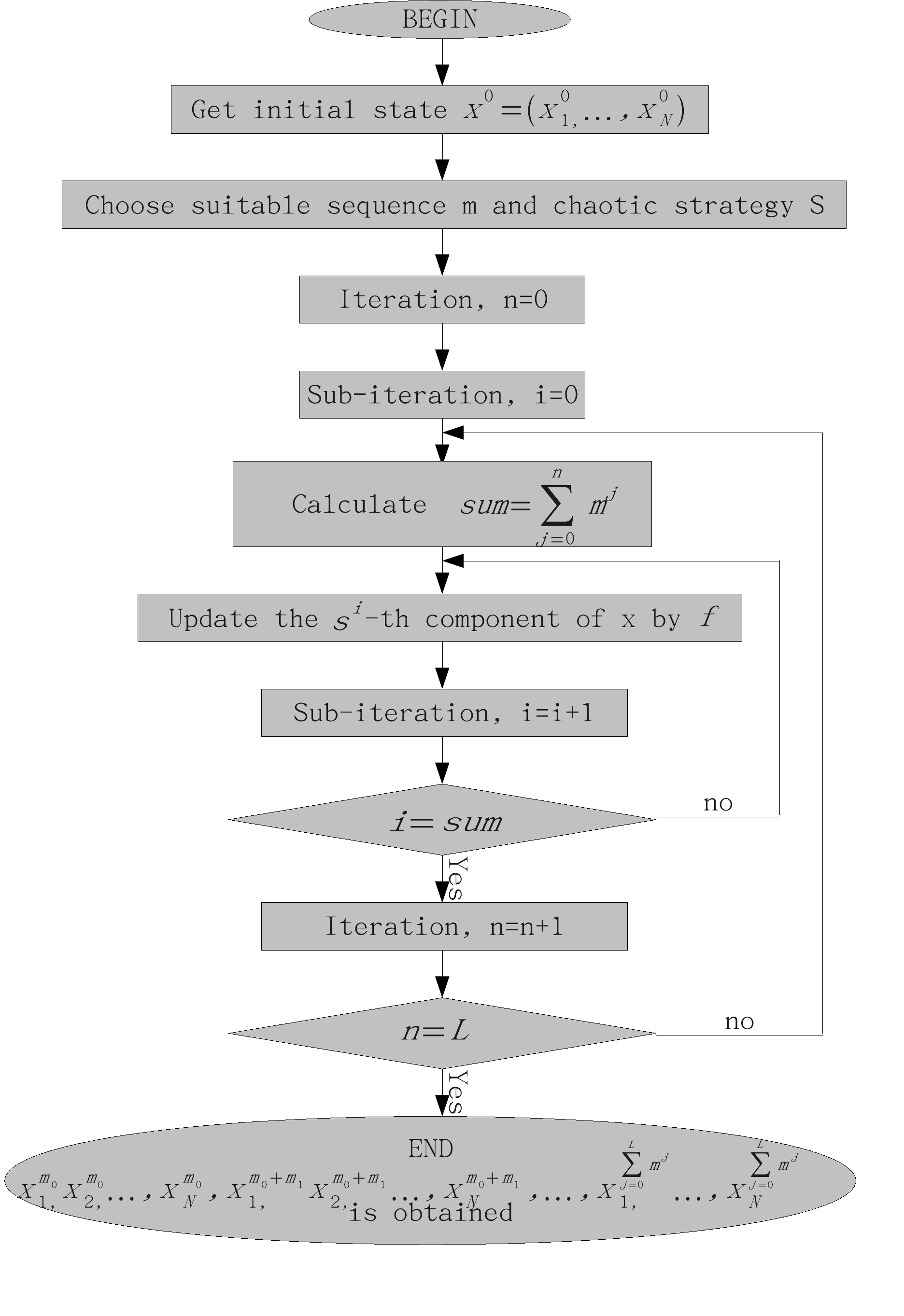}
\DeclareGraphicsExtensions.
\caption{Flow chart of CIPRNG version 1}
\label{Flow chart of chaotic strategy}
\end{figure}

Some chaotic iterations are fulfilled (see Flow chart~\ref{Flow chart of chaotic strategy}) 
to generate a sequence 
$\left(x^n\right)_{n\in\mathds{N}} \in \left(\mathds{B}^N\right)^\mathds{N}$ of Boolean vectors, 
which are the successive states of the iterated system. Some of these vectors are randomly extracted 
and their components constitute the pseudorandom bit flow~\cite{DBLP:journals/corr/abs-1004-2374}.
Chaotic iterations are realized as follows. The initial state $x^0 \in \mathds{B}^N$ is a 
Boolean vector taken as a seed and the chaotic strategy 
$\left(S^n\right)_{n\in\mathds{N}}\in \llbracket 1, N \rrbracket^\mathds{N}$ is constructed with PRNG2.
At each iteration, only the $S^i$-th component of state 
$x^n$ is updated. Finally, some $x^n$ are selected by a sequence $m^n$, obtained
using the PRNG1, as the pseudorandom bit sequence of our generator.

\begin{tiny}
\begin{table*}[!t]
\renewcommand{\arraystretch}{1.3}
\caption{Running example of CIPRNG version 1}
\label{table application example}
\centering
  \begin{tabular}{|c|ccccc|cccccc|cccccc|}
    \hline
$m:$    &   &   & 4 &   &         &   &   & 5 &   &   &          &   &   & 4 &   &          & \\ \hline
$S$     & 2 & 4 & 2 & 2 &         & 5 & 1 & 1 & 5 & 5 &          & 3 & 2 & 3 & 3 &          & \\ \hline
$x^{0}$ &   &   &   &   & $x^{4}$ &   &   &   &   &   & $x^{9}$  &   &   &   &   & $x^{13}$ & \\
1       &   &   &   &   &
1       &   & $\xrightarrow{1} 0$ & $\xrightarrow{1} 1$ & & &
1      &   &   &   & &
1 & \\
0       & $\xrightarrow{2} 1$ & & $\xrightarrow{2} 0$ & $\xrightarrow{2} 1$ &
1       &   &   &   &   &  &
1       &   & $\xrightarrow{2} 0$ & & & 0 &\\
1       &   &   &   &   &
1       &   &   &   &   &  &
1       & $\xrightarrow{3} 0$ & & $\xrightarrow{3} 1$ & $\xrightarrow{3} 0$ &
0 &\\
0       &   & $\xrightarrow{4} 1$  &   & &
1       &   &   &   &   &  &
1       &   &   &   &   &
1 &\\
0       &   &   &   &   &
0       & $\xrightarrow{5} 1$ &   &  & $\xrightarrow{5} 0$ & $\xrightarrow{5} 1$ &
1      &   &   &   &   &
1 &\\
\hline
  \end{tabular}\\
\vspace{0.5cm}
Output: $x_1^{0}x_2^{0}x_3^{0}x_4^{0}x_5^{0}x_1^{4}x_2^{4}x_3^{4}x_4^{4}x_5^{4}x_1^{9}x_2^{9}x_3^{9}x_4^{9}$
$x_5^{9}x_1^{13}x_2^{13}x_3^{13}x_4^{13}x_5^{13}... = 10100111101111110011...$
\end{table*}
\end{tiny}

The basic design procedure of the first version of the CIPRNG generator 
is summed up in Algorithm \ref{version 1 ci}.
The internal state is $x$, whereas $a$ and $b$ are computed by PRNG1 and PRNG2.
See Table~\ref{Chaotic iteration}
for a run example of this CIPRNG version 1.

\begin{algorithm}
\textbf{Input:} the internal state $x$ (an array of $\mathsf{N}$ 1-bit words)\\
\textbf{Output:} an array $r$ of $\mathsf{N}$ 1-bit words
\begin{algorithmic}[1]

\STATE$a\leftarrow{PRNG1()}$;
\STATE$m\leftarrow{a~mod~2+c}$;
\WHILE{$i=0,\dots,m$}
\STATE$b\leftarrow{PRNG2()}$;
\STATE$S\leftarrow{b~mod~\mathsf{N}}$;
\STATE$x_S\leftarrow{ \overline{x_S}}$;
\ENDWHILE
\STATE$r\leftarrow{x}$;
\STATE return $r$;
\medskip
\caption{An arbitrary round of the CIPRNG Version 1}
\label{version 1 ci}
\label{Chaotic iteration}
\end{algorithmic}
\end{algorithm}

\subsubsection{CIPRNG(PRNG1,PRNG2): Version 2 }
The second version of the CI-based generators is designed by the following process~\cite{bfgw11:ij}. First of all, some chaotic iterations have to be done to generate a sequence $\left(x^n\right)_{n\in\mathds{N}} \in \left(\mathds{B}^{32}\right)^\mathds{N}$ of Boolean vectors, which are the successive states of the iterated system. Some of these vectors will be randomly extracted and the pseudorandom bit flow will be constituted by their components. Such chaotic iterations are realized as follows. 
\begin{itemize}
\item Initial state $x^0 \in \mathds{B}^{N}$ is a Boolean vector taken as a seed.
\item Chaotic strategy $\left(S^n\right)_{n\in\mathds{N}}\in \llbracket 1,N \rrbracket^\mathds{N}$ is
an irregular decimation of the PRNG2 sequence.
\end{itemize}
At each iteration, only the $S^i$-th component of state $x^n$ is updated using the vectorial
negation, as follows: $x_i^n = x_i^{n-1}$ if $i \neq S^i$, else $x_i^n = \overline{x_i^{n-1}}$.
Finally, some $x^n$ are selected by a sequence $m^n$ as the pseudorandom bit sequence of our generator, where
$(m^n)_{n \in \mathds{N}} \in \mathcal{M}^\mathds{N}$ is computed from PRNG1.
\begin{algorithm}
\textbf{Input:} the internal state $x$ ($\mathsf{N}$ bits)\\
\textbf{Output:} a state $r$ of $\mathsf{N}$ bits
\begin{algorithmic}[1]
\FOR{$i=0,\dots,N$}
{
\STATE$d_i\leftarrow{0}$\;
}
\ENDFOR
\STATE$a\leftarrow{PRNG1()}$\;
\STATE$m\leftarrow{g_1(a)}$\;
\STATE$k\leftarrow{m}$\;
\WHILE{$i=0,\dots,k$}

\STATE$b\leftarrow{PRNG2()~mod~\mathsf{N}}$\;
\STATE$S\leftarrow{b}$\;
    \IF{$d_S=0$}
    {
\STATE      $x_S\leftarrow{ \overline{x_S}}$\;
\STATE      $d_S\leftarrow{1}$\;
    
    }
    \ELSIF{$d_S=1$}
    {
\STATE      $k\leftarrow{ k+1}$\;
    }\ENDIF
\ENDWHILE
\STATE $r\leftarrow{x}$\;
\STATE return $r$\;
\medskip
\caption{An arbitrary round of the CIPRNG Version 2}
\label{version 2 ci}
\end{algorithmic}
\end{algorithm}
The basic design procedure of this CIPRNG Version 2 generator is summarized in Algorithm \ref{version 2 ci}. The internal state is $x$. $a$ and $b$ are those computed by the two inputted PRNGs. Finally, the value $m$ is the integers sequence defined in Eq.(\ref{Formula}).

\begin{equation}
\label{Formula}
m^n = g_1(S^n)=
\left\{
\begin{array}{l}
0 \text{ if }0 \leqslant{S^n}<{C^0_{32}},\\
1 \text{ if }{C^0_{32}} \leqslant{S^n}<\sum_{i=0}^1{C^i_{32}},\\
2 \text{ if }\sum_{i=0}^1{C^i_{32}} \leqslant{S^n}<\sum_{i=0}^2{C^i_{32}},\\
\vdots~~~~~ ~~\vdots~~~ ~~~~\\
N \text{ if }\sum_{i=0}^{N-1}{C^i_{32}}\leqslant{S^n}<1.\\
\end{array}
\right.
\end{equation}

\section{Security Analysis of CIPRNG Version 1}
\label{Security Analysis Version 1 CI}
In this section the concatenation of two strings $u$ and $v$ is classically denoted by $uv$. In a cryptographic context, a pseudorandom generator is a deterministic algorithm $G$ transforming strings into strings and such that, for any seed $s$ of length m, $G(s)$ (the output of $G$ on the input $s$) has size $l_G(m)$ with $l_G(m) > m$. The notion of secure PRNGs can now be defined as follows. 

\subsection{Algorithm expression conversion}
For the convenience of security analysis, CIPRNG Version 1 detailed in Algorithm \ref{Chaotic iteration} is converted as in Eq.(\ref{Version 1 CI Eq}), where internal state is $x$, $S$ and $T$ are those computed by PRNG1 and PRNG2, whereas at each round, $x^{n-1}$ is updated to $x^n$. 

\begin{equation}
\left\{
\begin{array}{l}
x^0 \in \llbracket 0, 2^\mathsf{N}-1 \rrbracket, S \in \llbracket 0, 2^\mathsf{N}-1 \rrbracket^\mathds{N}, T \in \llbracket 0, 2^\mathsf{N}-1 \rrbracket^\mathds{N}\\
C = S^n \& 1 + 3*N\\
w^0 = {T}^m~mod~N, w^1 = {T}^{m+1} \& 3, ... w^{C-1} = {T}^{m+C-1} \& 3\\ 
d^n = (1 \ll w^0) \oplus (1\ll w^1) \oplus ... (1 \ll w^{C-1})\\
\forall n \in \mathds{N}^*, x^n = x^{n-1} \oplus d^n.
\end{array}
\right.
\label{Version 1 CI Eq}
\end{equation}

\subsection{Security notion}
\label{proof1}
\begin{definition}
\label{CSPRNG}
A cryptographic PRNG $G$ is secure if for any probabilistic polynomial time algorithm D, for any polynomial p, and for all sufficiently large m's,  
\begin{equation}
|Pr[D(G(U_m))=1]-Pr[D(U_{l_G(m)})=1]<\frac{1}{p(m)},
\end{equation}
where $U_r$ is the uniform distribution over $\{0, 1\}^r$ and the probabilities are taken over $U_m$, $U_{l_G(m)}$ as well as over the internal coin tosses of $D$.
\end{definition}

Intuitively, it means that there is no polynomial time algorithm that can distinguish a perfect uniform random generator from $G$ with a non negligible probability. Note that it is quite easily possible to change the function $l$ into any polynomial function $l'$ satisfying $l'(m)>m$.

The generation schema developed in Eq.\ref{Version 1 CI Eq} is based on two pseudorandom generators. Let $H$ be the ``PRNG1'' and $I$ be the ``PRNG2''. We may assume, without loss of generality, that for any string $S_0$ of size $L$, the size of $H(S_0)$ is $kL$,  then for any string $T_0$ of size $M$, it has $I(T_0)$ with $kN$, $k > 2$. It means that $l_H(N) = kL$ and $l_I(N) = kM$. Let $S_1,...,S_k$ be the string of length $L$ such that $H(S_0) = S_1 ... S_k$ and $T_1,...,T_k$ be the string of length $M$ s.t. $H(S_0) = T_1 ... T_k$ ($H(S_0)$ and $I(T_0)$ are the concatenations of $S_i$'s and $T_i$'s). 

The generator $X$ defined in Algorithm \ref{Version 1 CI Eq} is mapping any string $x_0S_0T_0$, of length $L+M+N$, into the string $x_0 \oplus d^1, x_0 \oplus d^1 \oplus d^2,...(x_0 \bigoplus^{i=k}_{i=0}d^i)$, c.f. Eq.(\ref{Version 1 CI Eq}). One in particular has $l_X(L+M+N) = kN = l_H(N)$ and $k > M+L+N$. We announce that if the inputted generator $H$ is cryptographically secure, then the new one defined in Eq.(\ref{Version 
1 
CI Eq}) is secured too.

\begin{proposition}
\label{cryptopreuve}
If PRNG1 is a secure cryptographic generator, then for all PRNG2, we can have that $X$ is a secure cryptographic
PRNG too.
\end{proposition}

\begin{proof}
The proposition is proven by contraposition. Assume that $X$ is not
secure. By definition, there exists a polynomial time probabilistic
algorithm $D$, a positive polynomial $p$, such that for all $k_0$ there exists
$L+M+N\geq {k_0}$ satisfying 
$$| \mathrm{Pr}[D(X(U_{L+M+N}))=1]-\mathrm{Pr}[D(U_{kN}=1)]|\geq \frac{1}{p(L+M+N)}.$$

Consider a word $w$ of size $kL$.
\begin{enumerate}
 \item Decompose $w$ into $w = w_1...w_k$.
 \item Pick a string $y$ of size $N$ uniformly at random.
 \item Pick a string of size $(3kN + \sum_{j=1}^{j=k}(w_j\&1)) M$: $u$.
 \item Decompose $u$ into $u = u_1...u_{3kN + \sum_{j=1}^{j=k}(w_j\&1)}$.
 \item Define $t_i = (\bigoplus_{l=3N(i-1)+(\sum_{l=1}^{l=i-1}(w_j\&1))+1}^
 {j=3N(i)+(\sum_{j=1}^{j=i}(w_j\&1))}(1<<u_l))$.
 \item Compute $z = (y\oplus t_1) (y\oplus t_1 \oplus t_2) ... (y\bigoplus_{i=1}^{i=k}(t_i))$.
 \item Return $D(z)$.
\end{enumerate}

On one hand, consider for each $y\in \mathbb{B}^{kN}$ the function $\varphi_{y}$ from $\mathbb{B}^{kN}$ into $\mathbb{B}^{kN}$ mapping $t=t_1\ldots t_k$ (each $t_i$ has length $N$) to $(y\oplus t_1 )(y\oplus t_1\oplus t_2)\ldots (y \bigoplus_{i=1}^{i=k} t_i)$.  On the other hand, treat each $u_l \in \mathbb{B}^{(3Nk + \sum_{j=0}^{j=k}(w_j\&1)) M}$ by the function $\phi_{u}$ from $\mathbb{B}^{(3kN + \sum_{j=0}^{j=k}(w_i\&1)) M}$ into $\mathbb{B}^{kN}$ mapping $w = w_1 \ldots w_k$ (each $w_i$ has length $L$) to: \\ $(\bigoplus_{l=1}^{l=3N+(w_1\&1)}(1<<u_l))  ((\bigoplus_{l=1+3N+(w_1\&1)}^{l=6N+(w_1\&1)+(w_1\&1)}(1<<u_l))$\\  $ \ldots (\bigoplus_{l=3N(k-1)+\sum_{j=1}^{j=k-1}(w_j\&1)}^{l=3Nk+\sum_{j=1}^{j=k}(w_j\&1)}(1<<u_l))$.\\ 
By construction, one has for every $w$,
\begin{equation}\label{PCH-1}
D^\prime(w)=D(\varphi_y(\phi_u(w))).
\end{equation}

Therefore, and using Eq.(\ref{PCH-1}),
one has\\
$\mathrm{Pr}[D^\prime(U_{kL})=1]=\mathrm{Pr}[D(\varphi_y(\phi_u(U_{kL})))=1]$ \\ and,
therefore, 
\begin{equation}\label{PCH-2}
\mathrm{Pr}[D^\prime(U_{kL})=1]=\mathrm{Pr}[D(U_{kN})=1].
\end{equation}

Now, using Eq.(\ref{PCH-1}) again, one has  for every $x$,
\begin{equation}\label{PCH-3}
\mathrm{Pr}[D^\prime(U_{H(x)})=1]=\mathrm{Pr}[D(\varphi_y(\phi_u(U_{H(x)})))=1].
\end{equation}

Since where $y$ and $u_j$ are randomly generated. By construction, $\varphi_y(\phi_u(x))=X(yu_1w)$, hence 

\begin{equation}\label{PCH-4}
\mathrm{Pr}[D^\prime(H(U_{kL}))=1]=\mathrm{Pr}[D(X(U_{N+M+L}))=1].
\end{equation}

Compute the difference of Eq.(\ref{PCH-4}) and Eq.(\ref{PCH-3}), one can deduce that there exists a polynomial time probabilistic algorithm $D^\prime$, a positive polynomial $p$, such that for all $k_0$ there exists $L+M+N\geq {k_0}$ satisfying $$| \mathrm{Pr}[D^\prime(H(U_{KL}))=1]-\mathrm{Pr}[D(U_{kL})=1]|\geq \frac{1}{p(L+M+N)},$$ proving that $H$ is not secure, which is a contradiction. 
\end{proof}

Compared to stream ciphers, which are symmetric key ciphers where plaintext digits are combined with 
a pseudorandom cipher digit stream (keystream), the CIPRNG method can be described as a post-treatment
on two inputted PRNGs, that:
\begin{enumerate}
 \item add chaotic properties to these generators,
 \item by doing so, improve their statistical properties when the inputs are defective,
 \item while preserving their security, for instance when one of the input is cryptographically secure.
\end{enumerate}
If PRNG1 is already used as a keystream in a stream cipher, because it is cryptographically secure,
then the combined CIPRNG(PRNG1,XORshift), which runs potentially faster
than PRNG1, can be used too as a keystream. The security comparison between CIPRNG and other designs 
is thus summarized in
Proposition~\ref{cryptopreuve}: the security of\linebreak CIPRNG(PRNG1,PRNG2) is directly related to the
one of PRNG1, meaning that if PRNG1 is secure, then the resulted CIPRNG is secure too.

\section{CIPRNG Version 1 Designed for FPGA}
\subsection{An efficient and cryptographically secure PRNG based on CIPRNG  Version 1}
In Algorithm \ref{Version 1 CI cs} is given an efficient and cryptographically secure generator
suitable for FPGA applications. It is based on CIPRNG Version 1 and thus presents a good random 
statistical profile.

\begin{algorithm}
\textbf{Notice}: xorshift1, xorshift2 (64-bit XORshift generators)\\
\textbf{Input}: $z$ (a 16-bit word)\\
\textbf{Output}: $r$ (a 16-bit word)\\
\begin{algorithmic}[1]
\STATE$x \leftarrow xorshift1();$\\
\STATE$y \leftarrow xorshift2();$\\
\STATE$z1 \leftarrow x \& 0xffffffff$\\
\STATE$z2 \leftarrow (x >> 32) \& 0xffffffff$\\
\STATE$z3 \leftarrow y \& 0xffffffff$\\
\STATE$z4 \leftarrow (y >> 32) \& 0xffffffff$\\
\STATE$t \leftarrow bbs();$\\
\STATE$t1 \leftarrow t \& 1;$\\
\STATE$t2 \leftarrow t \& 2;$\\
\STATE$t3 \leftarrow t \& 4;$\\
\STATE$t4 \leftarrow t \& 8;$\\
\STATE$w1 \leftarrow 0;$\\
\STATE$w2 \leftarrow 0;$\\
\STATE$w3 \leftarrow 0;$\\
\STATE$w4 \leftarrow 0;$\\
\WHILE{$i=0,\dots,11$}
\STATE$w1 \leftarrow (w1 \oplus (1 \ll ((z1 \gg (i\times 2))\&3)));$\\
\STATE$w2 \leftarrow (w2 \oplus (1 \ll ((z2 \gg (i\times 2))\&3)));$\\
\STATE$w3 \leftarrow (w3 \oplus (1 \ll ((z3 \gg (i\times 2))\&3)));$\\
\STATE$w4 \leftarrow (w4 \oplus (1 \ll ((z4 \gg (i\times 2))\&3)));$\\
\ENDWHILE
\STATE$\textbf{if}~(t1 \neq 0)~\textbf{then}~w1 \leftarrow (w1 \oplus (1 \ll ((z1 \gg 24)\&3)));$ \\
\STATE$\textbf{if}~(t2 \neq 0)~\textbf{then}~w2 \leftarrow (w2 \oplus (1 \ll ((z2 \gg 24)\&3)));$ \\
\STATE$\textbf{if}~(t3 \neq 0)~\textbf{then}~w3 \leftarrow (w3 \oplus (1 \ll ((z3 \gg 24)\&3)));$ \\
\STATE$\textbf{if}~(t4 \neq 0)~\textbf{then}~w4 \leftarrow (w4 \oplus (1 \ll ((z4 \gg 24)\&3)));$ \\
\STATE$z \leftarrow z \oplus w1 \oplus (w2 \ll 4) \oplus (w3 \ll 8) \oplus (w4 \ll 12);$\\
\STATE$r \leftarrow z;$\\
\STATE$\textbf{return} ~r;$\\
\end{algorithmic}
\caption{An efficient and cryptographically secure generator based on CIPRNG version 1}
\label{Version 1 CI cs}
\end{algorithm}

The internal state $x$ is a vector of $16$ bits, whereas two $64$-bit XORshift generators 
($xorshift1(), xorshift2()$) are provided as entropy sources. As it can be seen in the algorithm, 
the two outputs of XORshift generators are spread into four $32$-bit integers. Then for each integer, there are $16$ $2-$bits components that can be found; every $12$ of these components are used to update the states. Lastly, the $4$ least significant bits (LSBs) of the output $bbs()$ of the Blum Blum Shub generator
decide if the state must be updated with the considered $13$-bits block or not. 

According to Section~\ref{Security Analysis Version 1 CI}, this generator based on CIPRNG version 1 can turn to be cryptographically secure, if the PRNG1 entropy source is cryptographically secure. Here, this inputted generator is the well known BBS, which is believed to be the most secured PRNG method currently available~\cite{vmd}. The $t$ value is computed by a BBS with a modulo $m$ equal to $32$ bits. Then the $log(log(m))$ LSBs of $t$ can be treated as secure, this is why we only considerate $4$ LSBs in this algorithm.

\begin{figure*}
\begin{center}
  \subfigure[XORshift]{\includegraphics[width=5.5cm]{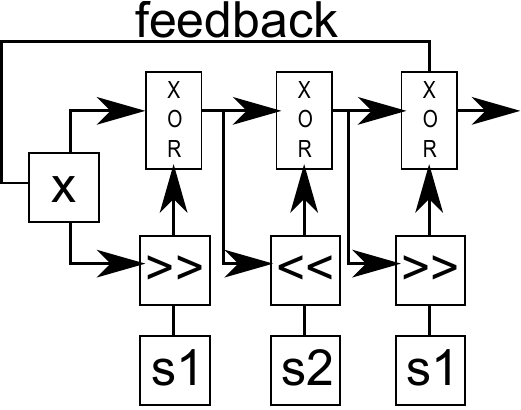}
  \label{xorshift verilog}}
  \subfigure[BBS]{\includegraphics[width=5.5cm]{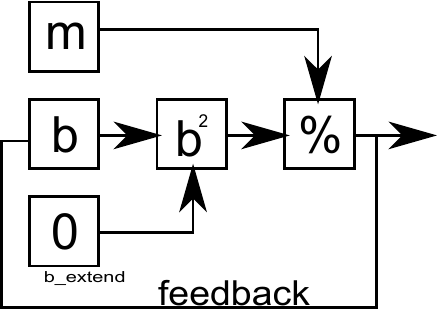}
  \label{BBS verilog}}
  \subfigure[The proposed CIPRNG]{\includegraphics[width=15cm]{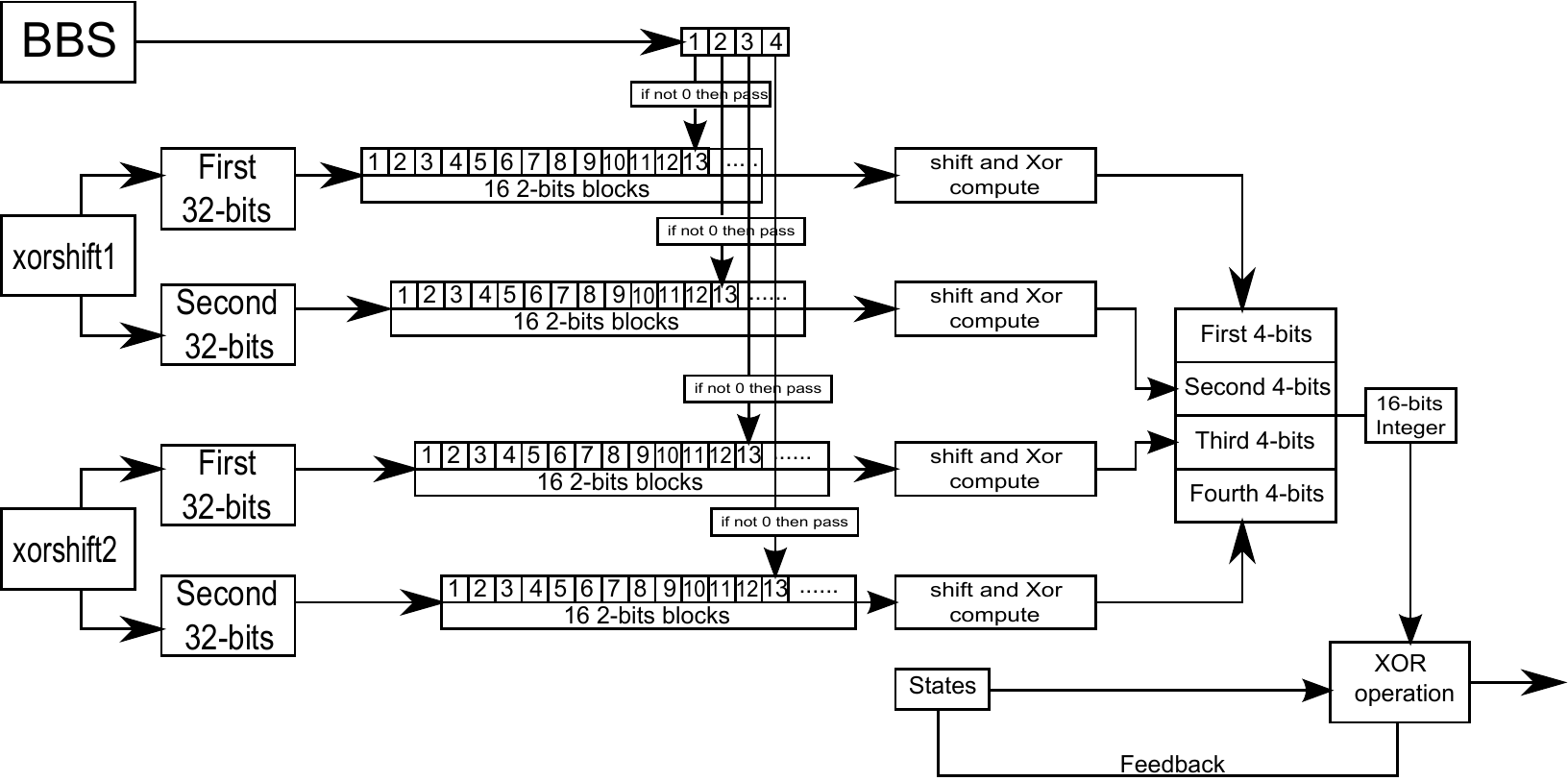}
  \label{CI verilog}}
\end{center}
\caption{The processing structure for BBS in FPGA (per clock step)}
\end{figure*}
Following the approach detailed in~\cite{bfg12a:ip}, we thus have used chaotic iterations in order to improve the statistical behavior of the inputted generators.
Here, two coupled $64$ bits XORshift generators together with one BBS are applied. By doing so, we obtain in Algorithm \ref{Version 1 CI cs} a generator being both chaotic and cryptographically secure~\cite{DBLP:journals/corr/abs-1112-5239}.

\begin{table*}
  \renewcommand{\arraystretch}{1.3}
  \caption{NIST SP 800-22 test results ($\mathbb{P}_T$)}
  \label{nist}
  \centering
  \begin{tabular}{|l||c|c|c|}
    \hline
    Method &CIPRNG & XORshift & BBS  \\ \hline\hline
    Frequency (Monobit) Test 			& 0.073128 &  0.145326 & 0.32435 \\ \hline
    Frequency Test within a Block  			& 0.719128 &  0.028817& 0.000000 \\ \hline
    Runs Test 					& 0.314992 &  0.739918 & 0.000000 \\ \hline
    Longest Run of Ones in a Block Test 		& 0.445121 &  0.554420 & 0.000000  \\ \hline
    Binary Matrix Rank Test 			& 0.888124 &  0.236810 & 0.000000\\ \hline
    Discrete Fourier Transform (Spectral) Test	& 0.912003 &  0.514124 & 0.000000 \\ \hline
    Non-overlapping Template Matching Test* 	& 0.500459 &  0.512363 & 0.000000\\ \hline
    Overlapping Template Matching Test   		& 0.702445 &  0.595549 & 0.000000 \\ \hline
    Universal Statistical Test   	& 0.666230 &  0.122325 & 0.000000\\ \hline
    Linear Complexity Test  			& 0.475761 &  0.249284  & 0.000000\\ \hline
    Serial Test* (m=10) 				& 0.780099 &  0.495847 &  0.043355\\ \hline
    Approximate Entropy Test (m=10) 		& 0.679102 &  0.000000 &  0.000000 \\ \hline
    Cumulative Sums (Cusum) Test* 			& 0.819200 &  0.074404 & 0.000000\\ \hline
    Random Excursions Test* 			& 0.697803 &  0.507812 & 0.000000 \\ \hline
    Random Excursions Variant Test* 		& 0.338243 &  0.289594 &  0.000000  \\ \hline
    Success & 15/15& 14/15 & 2/15 \\ \hline
  \end{tabular}
\end{table*}

Table~\ref{nist} shows the test results of the proposed CIPRNG against the NIST battery~\cite{ANDREW2008}. Results of XORshift and BBS are provided too. 
According to NIST test suite, the sole BBS generator algorithm cannot produce a statistically perfect output. This is not contradictory with Prop.~\ref{cryptopreuve}, as the cryptographically secure property is an asymptotic one: 
even though the Blum Blum Shum generator is cryptographically secure (which is a  property independent from the chosen modulo $m$), the very small value chosen for $m$ makes it unable to pass the NIST battery. Obviously, best statistical performances are obtained using the proposed CIPRNG.

\subsection{FPGA Design}
\label{FPGA design}
In order to take benefits from the computing power of FPGA, a whole processing needs to spread into several independent blocks of threads that can be computed simultaneously. In general, the larger the number of threads is, the more logistic elements of FPGA are used, and the less branching  instructions are used  (if,  while,  ...),  the  better the performances on FPGA are. Obviously, having these requirements in  mind, it is possible to build a program similar to the algorithm presented in Algorithm \ref{Version 1 CI cs}, which produces pseudorandom numbers with chaotic properties on FPGA. To do so, Verilog-HDL~\cite{verilog} has been used to help programming. In this generator, there are three PRNG objects that use the exclusive or operation, two XORshifts, and a BBS, their processing are described thereafter.

\subsubsection{Design of XORshift}
The structure of XORshift designed in Verilog-HDL is shown in Figure~\ref{xorshift verilog}. There are four inputs:
\begin{itemize}
\item The first one is the initial state, which costs 64 bits 
of register units,
\item the other three ones are used to define the shift operations.
\end{itemize}
Let us remark that, in FPGA, this shift operation costs nothing, as it simply consists in using different bit cells of the input. We can thus conclude that there are $64 - s1 + 64 -s2 + 64 -s3 = 192 - s1 - s2 - s3$ logic gates elements that are required for the XORshifts processing. 

\subsubsection{Design of BBS}
Figure~\ref{BBS verilog} gives the proposed design of the BBS generator in FPGAs. There are two inputs of $32$ bits, namely $b$ and $m$. Register $b$ stores the state of the system at each time (after the square computation). $m$ is also a register that saves the value of $M$, which must not change. Another register $b\_extend$ is used to combine $b$ to a data having $64$ bits, with a view to avoid overflow. After the last computation, the three LSBs from the output of $\%$ are taken as output. Let us notice that a BBS is performed at each time unit.

\subsubsection{Design of CI}
Two XORshifts and one BBS are connected to work together, in order to compose the proposed CIPRNG (see Figure~\ref{CI verilog}). 
As it can be shown, the four bits of the BBS output are switches for the corresponding $32$ bits outputs from XORshift. Every round of the processing costs two time units to be performed: in the first clock, the three PRNGs are processed in parallel, whereas in the second one, the results of these generators are combined with the current state of the system, in order to produce the output of $16$ bits.

\begin{figure*}
\begin{center}
  \includegraphics[width=11cm]{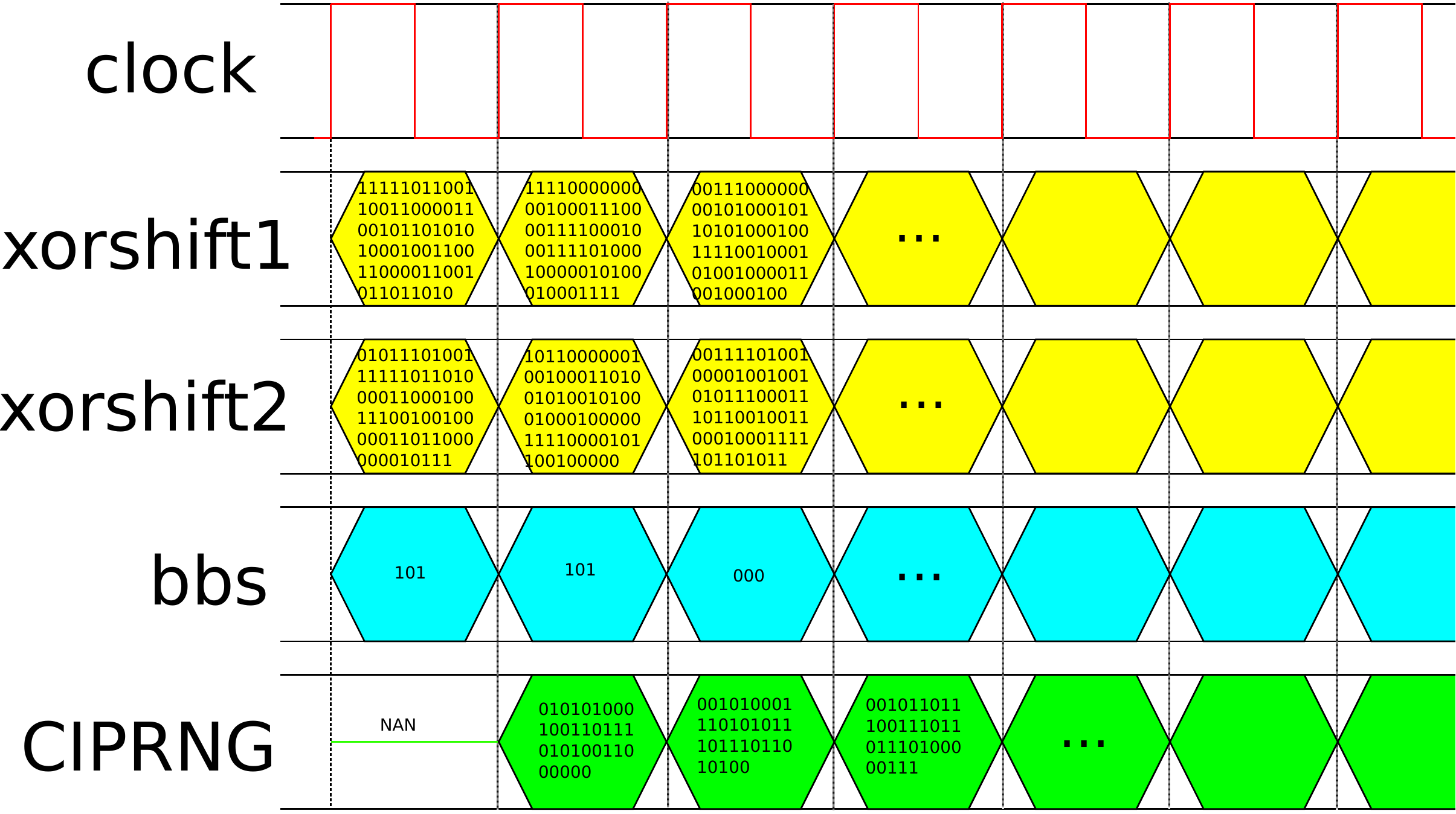}
\end{center}
\caption{Outputs of each component in clock step unit}
 \label{work_flow}
\end{figure*} 
In our experiments, the type $EP2C8Q208C8$ from Altera company's CYCLONE II FPGA series has been used. By default, its working frequency is equal to $50$ MHz. However, it is possible to increase it until $400$ MHz by using the phase-lock loop (PLL) device. In that situation, the CIPRNG designed on this FPGA can produce over $6000$ Mbits per second (that is, $400 (MHz) \times 16 (bits)$, see Figure~\ref{work_flow}), while using $6114$ of the $8256$ logic elements in $EP2C8Q208C8$. This is nearly $30$ times faster than when it is processed in continuous method. 

In the next section, an application of this CSPRNG designed on FPGA in the information hiding security fields is detailed, to show that this hardware pseudorandom generator is ready to use.

\section{An Information Hiding Application}
\label{application}

Information hiding has recently become a major information security technology, especially with the increasing importance and widespread distribution of digital media th-rough the Internet \cite{Wu2007bis}. It includes several techniques like digital watermarking. The aim of digital watermarking is to embed a piece of information into digital documents, such as pictures or movies. This is for a large panel of reasons, such as: copyright protection, control utilization, data description, content authentication, and data integrity. For these reasons, many different watermarking schemes have been proposed in recent years. Digital watermarking must have essential characteristics, including: security, imperceptibility, and robustness. Chaotic methods have been proposed to encrypt the watermark before embedding it in the carrier image for these security reasons. In this paper, a watermarking algorithm based on the chaotic PRNG presented above is given, as an illustration of use of this PRNG based on CI.

\subsection{Most and least significant coefficients}
The definitions of most and least significant coefficients are shown at first, as they have been formerly introduced in~\cite{gfb10:ip,bg10:ip}.

\begin{definition}
For a given image, the most significant coefficients (in short MSCs), are
coefficients that allow the description of the relevant part of the image,
\emph{i.e.}, its most rich part (in terms of embedding information), through a sequence of bits.
\end{definition}

\begin{definition}
By least significant coefficients (LSCs), we mean a translation of some
insignificant parts of a medium in a sequence of bits (insignificant can be understand as: 
``which can be altered without sensitive damages'').
\end{definition}

These LSCs can be for example, the last three bits of the gray level of each pixel, in the case of a spatial domain watermarking of a gray-scale image.

In the proposed application, LSCs are used during the embedding stage: some of the least significant coefficients of the carrier image will be chaotically chosen and replaced by the bits of the mixed watermark. With a large number of LSCs, the watermark can be inserted more than once and thus the embedding will be more secure and robust, but also more detectable. The MSCs are only useful in the case of authentication: encryption and embedding stages depend on them. Hence, a coefficient should not be defined at the same time, as a MSC and a LSC; the last can be altered, while the first is needed to extract the watermark. For a more rigorous definition of such LSCs and MSCs see, \emph{e.g.},~\cite{bcg11:ij}.

\subsection{Stages of the algorithm}

We recall now a formerly introduced watermarking scheme, which consists of two stages: (1) mixture of the watermark and (2) its embedding~\cite{bg10b:ip}.

\subsubsection{Watermark mixture}
Firstly, for safety reasons, the watermark can be mixed before its embedding into the image. A common way to achieve this stage is to use the bitwise exclusive or (XOR), for example, between the watermark and the above PRNG. In this paper and similarly to~\cite{bg10b:ip}, we will use another mixture scheme based on chaotic iterations. Its chaotic strategy, defined with our PRNG, will be highly sensitive to the MSCs, in the case of an authenticated watermark, as stated in ~\cite{guyeux:topological}.

\subsubsection{Watermark embedding}

Some LSCs will be substituted by all bits of the possibly mixed watermark. To choose the sequence of LSCs to be altered, a number of integers, less than or equal to the number $\mathsf{N}$ of LSCs corresponding to a chaotic sequence $\left( U^{k}\right) _{k}$, is generated from the chaotic strategy used in the mixture stage. Thus, the $U^{k}$-th least significant coefficient of the carrier image is substituted by the $k^{th}$ bit of the possibly mixed watermark. In the case of authentication, such a procedure leads to a choice of the LSCs that are highly dependent on the MSCs.

\subsubsection{Extraction}
The chaotic strategy can be regenerated, even in the case of an authenticated watermarking because the MSCs have not been changed during the stage of embedding the watermark. Thus, the few altered LSCs can be found, the mixed watermark can then be rebuilt, and the original watermark can be obtained. If the watermarked image is attacked, then the MSCs will change. Consequently, in the case of authentication and due to the high sensitivity of the embedding sequence, the LSCs designed to receive the watermark will be completely different. Hence, the result of the recovery will have no similarity with the original watermark: authentication is reached.

\begin{figure*}
\centering
\subfigure[General structure]{\includegraphics[width=6cm]{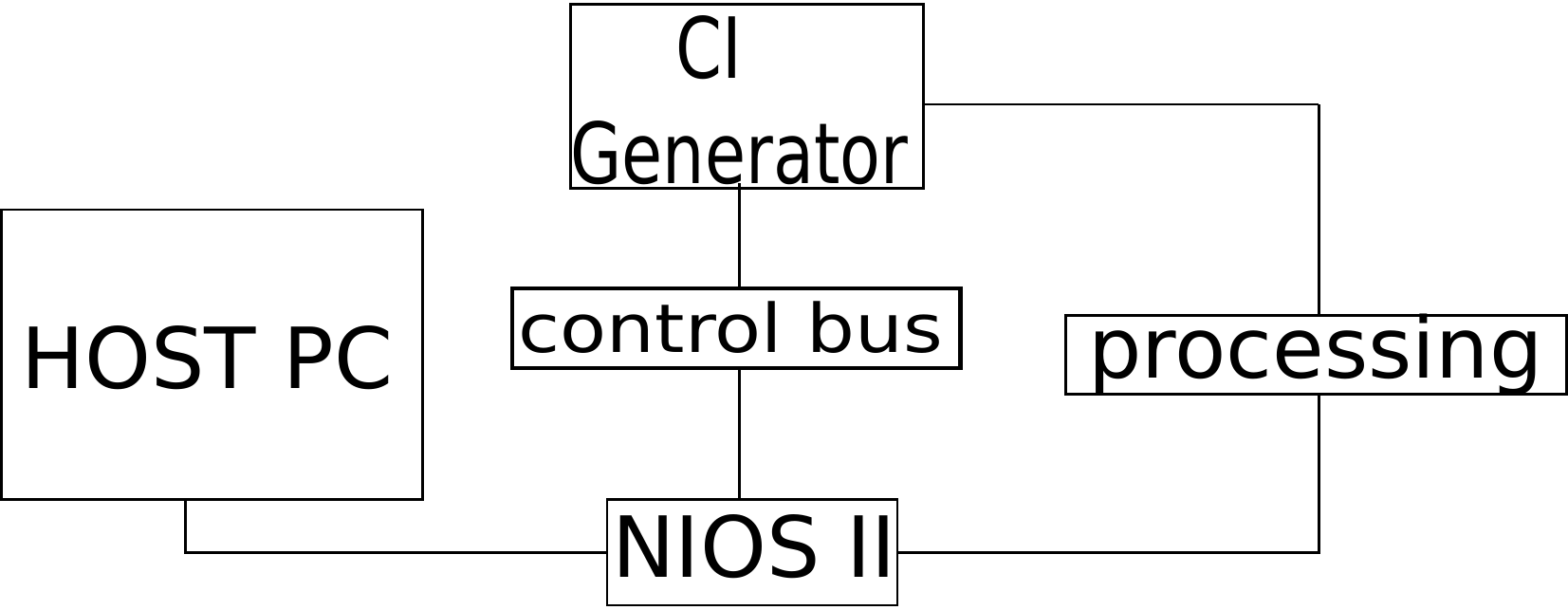}
\label{nios}} \hspace{0.5cm}
\subfigure[Schematic view]{\includegraphics[scale=0.4]{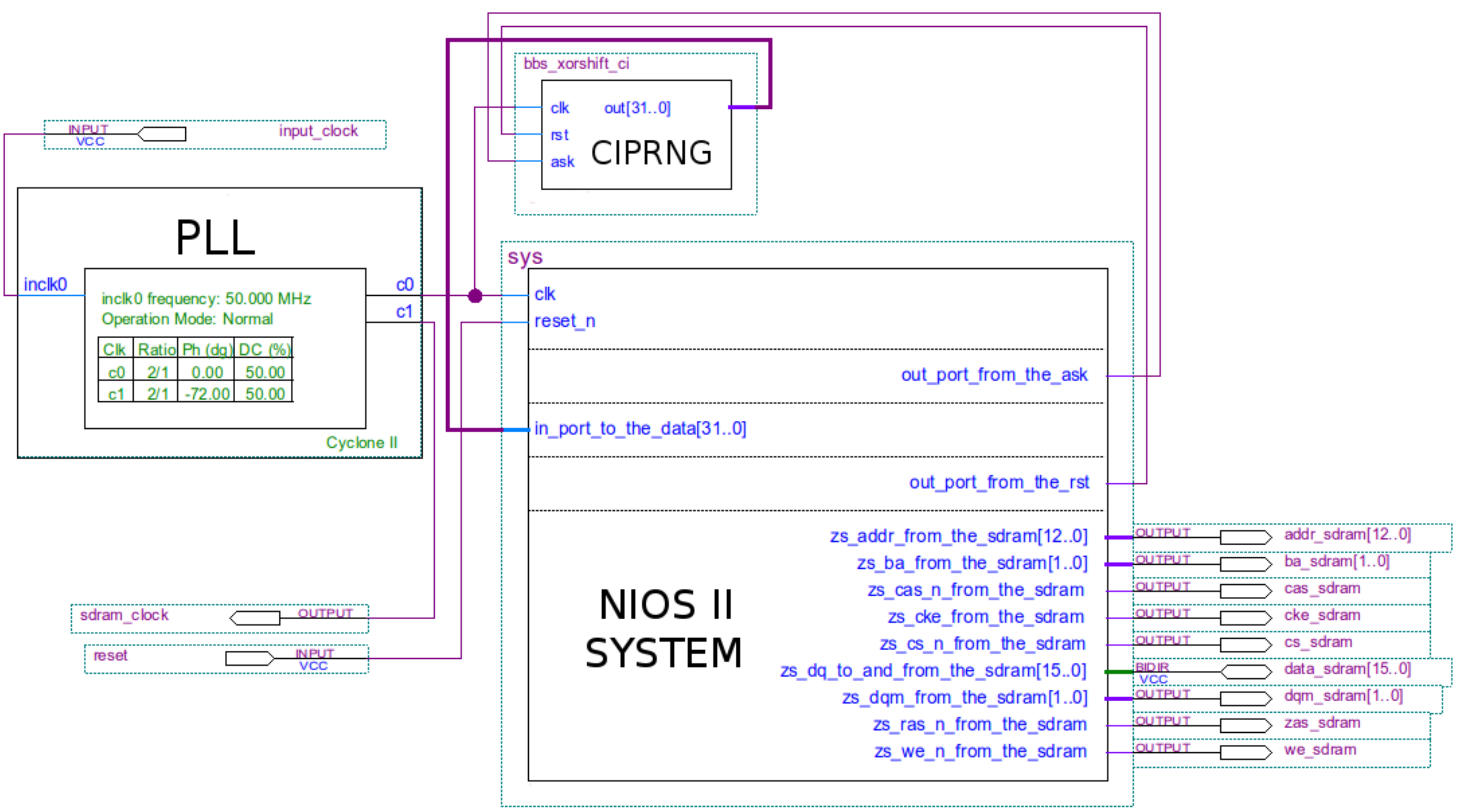}
\label{nios2}} \hspace{0.5cm}
\caption{NIOS II setting in FPGA}
\end{figure*}

\subsection{The FPGA setting}
The 32-bit embedded-processor architecture designed specifically for the Altera family of FPGAs 
is applied in this information hiding
specific application. Nios II incorporates many enhancements over the original 
Nios architecture, making it more suitable for a wider range of embedded computing applications, 
from DSP to system-control~\cite{nios}. 

Figure~\ref{nios} shows the structure of this application. 
The NIOS II system can read the image from the HOST computer side. 
Via the bus control, pseudorandom bits are 
produced into the FPGA and according to the CIPRNG. Then the results are transmitted back 
into the host. 

\begin{table*}
\caption{Robustness agains attacks}
\label{Cropping attacks}
\label{Rotation attacks}
\label{JPEG compression attacks}
\label{Gaussian noise attacks}
\begin{center}
{\footnotesize
\begin{tabular}{c c c c c}
\toprule
Attacks&\multicolumn{2}{c}{UNAUTHENTICATION} & \multicolumn{2}{c}{AUTHENTICATION}
\\ \midrule
\multirow{5}*{\rotatebox{90}{Cropping}} &Size (pixels) & Similarity & Size (pixels) & Similarity \\ \cmidrule(r){2-5}
		&10 & 99.18\% & 10 & 50.06\% \\
		&50 & 96.13\% & 50 & 54.44\% \\
		&100 & 91.21\% & 100 & 52.04\% \\
		&200 & 66.16\% & 200 & 50.88\% \\  \midrule
\multirow{5}*{\rotatebox{90}{Rotation } }		&Angle (degree) & Similarity & Angle (degree) & Similarity \\ \cmidrule(r){2-5}
		&2 & 96.11\% & 2 & 71.41\% \\
	&5 & 93.66\% & 5 & 60.03\% \\
		&10 & 92.55\% & 10 & 53.87\% \\
		&25 & 82.05\% & 25 & 50.09\% \\ \midrule
\multirow{5}*{\rotatebox{90}{JPEG compression}}	&&&&\\
	&Compression & Similarity & Compression & Similarity \\ \cmidrule(r){2-5}
		&2 & 81.90\% & 2 & 53.79\% \\
&5 & 66.43\% & 5 & 55.51\% \\
		&10 & 61.82\% & 10 & 51.24\% \\
		&20 & 54.17\% & 20 & 47.33\% \\
&&&&\\
\midrule
\multirow{5}*{\rotatebox{90}{Gaussian noise}}	\\

		&Standard dev. & Similarity & Standard dev. & Similarity \\ \cmidrule(r){2-5}
		&1 & 75.16\% & 1 & 51.05\% \\
&2 & 62.33\% & 2 & 50.35\% \\
		&3 & 56.34\% & 3 & 49.95\% \\ \bottomrule
\end{tabular}
}\\[0pt]
\end{center}
\end{table*}
In Figure~\ref{nios2}, the NIOS 
II is using the most powerful version the CYCLONE II can support (namely, the NIOS II/f one).
$4$ KB on chip memory and $16$ MB SDRAM are set, and the $PLL$ device is used to enhance the clock frequency from $50$ 
to $200$ MHz. Finally, the data connection bus NIOS II system and generator works in 32 bits.

\subsection{Results}
For evaluating the efficiency and the robustness of the application, some attacks are performed on some chaotically watermarked images. For the attacks, the similarity percentages with the original watermark are computed. These percentages are the numbers of equal bits between the original and the extracted watermark, shown as a percentage. A result less than or equal to $50\%$ implies that the image has probably not been watermarked.

\subsubsection{Cropping attack}
In this kind of attack, a watermarked image is cropped. In this case, the
results in Tab.\ref{Cropping attacks} have been obtained.
In Figure~\ref{fig:Cropping attack}, the decrypted watermarks are shown
after a crop of 50 pixels and after a crop of 10 pixels, in the
authentication case.
\begin{figure}[h!]
\centering
\subfigure[Unauthentication \newline ($10\times 10$)]{\includegraphics[width=2cm]{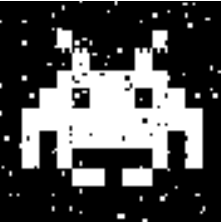}
\label{fig:watermark}} \hspace{0.5cm}
\subfigure[Authentication \newline ($10\times 10$)]{\includegraphics[width=2cm]{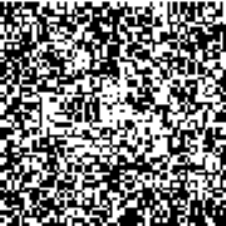}
\label{fig:LenaMarqueDwt}} \hspace{0.5cm}
\subfigure[Unauthentication \newline ($50\times 50$)]{\includegraphics[width=2cm]{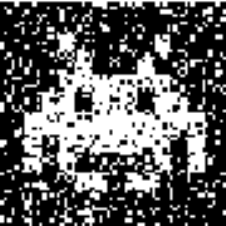}
\label{fig:LenaDiff}}
\caption{Extracted watermark after a cropping attack (zoom $\times 2$)}
\label{fig:Cropping attack}
\end{figure}

By analyzing the similarity percentage between the original and the
extracted watermark, we can conclude that in the case of unauthentication, the watermark still remains after a cropping attack. The desired robustness is reached. It can be noticed that cropping sizes and percentages are rather proportional. In the case of authentication, even a small change of the carrier image (a crop by $10\times 10$ pixels) leads to a really different extracted watermark. In this case, any attempt to alter the carrier image will be signaled, thus the image is well authenticated.

\subsubsection{Rotation attack}
Let $r_{\theta }$ be the rotation of angle $\theta $ around the center $%
(128, 128)$ of the carrier image. So, the transformation $r_{-\theta }\circ r_{\theta }$ is applied to the watermarked image. The results in Tab.\ref{Rotation attacks} have been obtained.
The same conclusion as above can be declaimed.

\subsubsection{JPEG compression}

A JPEG compression is applied to the watermarked image, depending on a
compression level. This attack leads to a change of
the representation domain (from spatial to DCT domain). In this case, the
results in Tab.\ref{JPEG compression attacks} have been obtained, illustrating a good authentication through JPEG attack.
As for the unauthentication case, the watermark still remains after a compression level equal to 10. This is a good result if we take into account the fact that we use spatial embedding.

\subsubsection{Gaussian noise}

A watermarked image can be also attacked by the addition of a Gaussian noise, depending on a standard deviation. In this case, the results in Tab.\ref{Gaussian noise attacks} are obtained, which are quite satisfactory
another time.

\subsection{Discussion}

Generally, the quality of a PRNG depends, to a large extent, on the following criteria: randomness, 
uniformity, independence, storage efficiency, and reproducibility. A chaotic sequence may satisfy these 
requirements and also other chaotic properties, as ergodicity, entropy, and expansivity. A chaotic 
sequence is extremely sensitive to the initial conditions. That is, even a minute difference in the 
initial state of the system can lead to enormous differences in the final state, even over fairly small 
timescales. Therefore, chaotic sequence fits the requirements of pseudorandom sequence well. Contrary to 
XORshift, our generator possesses these chaotic properties~\cite{guyeux09,wang2009}.
However, despite a large number of papers published in the field of chaos-based pseudorandom generators, 
the impact of this research is rather marginal. This is due to the following reasons: almost all PRNG 
algorithms using chaos are based on dynamical systems defined on continuous sets (\emph{e.g.}, the set
of real numbers). So these generators are usually slow, requiring considerably more storage space and lose 
their chaotic properties during computations. These major problems restrict their use as 
generators~\cite{Kocarev2001}.

In the CIPRNG method, we do not simply integrate chaotic maps hoping that the implemented algorithm remains 
chaotic. Indeed, the PRNG we conceive is just discrete chaotic iterations and we have proven 
in~\cite{guyeux09} that these iterations produce a topological chaos as defined by Devaney: 
they are regular, transitive, and sensitive to initial conditions. This famous definition of a 
chaotic behavior for a dynamical system implies unpredictability, mixture, sensitivity, and uniform 
repartition. Moreover, as only integers are manipulated in discrete chaotic iterations, the chaotic 
behavior of the system is preserved during computations, and these computations are fast.

These chaotic properties are behind the observed robustness of the proposed information hiding
scheme: transitivity, for instance, implies that the watermark is spread over the whole host
image, making it impossible to remove it by a simple crop. Regularity implies that the watermark
is potentially inserted several times, reinforcing the robustness obtained by topological mixing 
and transitivity. Expansivity and sensitivity guarantee us that authentication is reached, as in 
an authenticated watermarking, MSBs are taken into account, and even a slight alteration of these
bits leads to a completely different extracted watermark due to these metrical properties.
Finally, unpredictability plays obviously an important role in the security of the whole process againts
malicious attacks, even if this role is difficult to measure precisely in practice.

\section{Conclusion and future work}
\label{Conclusions and Future Work}

In this paper, the pseudorandom generator proposed in our former research work has been developed in terms of efficiency. We also have proven that this generator based on hardware can be cryptographically secure. By using a BBS generator and due to a new approach in the way the Version 1 CI PRNG uses its strategies, the generator based on chaotic iterations works faster and is more secure.
This new CIPRNG is able to pass NIST test suite when considering software implementation, and to reach $6000$ Mbps (with the throughtput is about $132/16$ each processing round) in FPGA hardware. These considerations enable us to claim that this CIPRNG(BBS, XORshift) offers a sufficient speed and level of security for a whole range of applications where secure generators are required as cryptography and information hiding. 

In future work, we will continue to explore new strategies and iteration functions. The chaotic behavior of the proposed generator
will be deepened by using the various tools provided by the mathematical theory of chaos.  
Additionally a probabilistic study of its security will be done. 
Lastly, new applications in computer science will be proposed, among other things in the Internet security field.

\bibliographystyle{plain} 
\bibliography{refs}
\end{document}